\documentclass[letterpaper,11pt]{article}
\usepackage{microtype}
\usepackage[letterpaper,margin=1in]{geometry}
\usepackage[numbers,sort&compress]{natbib}
\usepackage[table]{xcolor}
\usepackage{amsmath}
\usepackage{amssymb}
\usepackage{amsthm}
\usepackage{enumitem}
\usepackage{textgreek}
\usepackage{graphicx}
\usepackage{framed}
\usepackage[small]{caption}
\usepackage{booktabs}
\usepackage{titlesec}
\usepackage[unicode]{hyperref}
\usepackage{doi}

\titlespacing*{\paragraph}{0pt}{3ex plus 1.25ex minus .2ex}{1em}

\theoremstyle{plain}
\newtheorem{theorem}{Theorem}
\newtheorem{lemma}[theorem]{Lemma}
\newtheorem{property}[theorem]{Property}

\newtheorem*{conjecture*}{Conjecture}

\theoremstyle{definition}
\newtheorem{definition}[theorem]{Definition}

\theoremstyle{remark}

\newtheorem*{remark*}{Remark}

\newcommand{\dclocal}{{\upshape\sffamily DC\mbox{-}LOCAL}}
\newcommand{\local}{{\upshape\sffamily LOCAL}}
\newcommand{\lcl}{{\upshape\sffamily LCL}}
\newcommand{\tvirt}{\ensuremath{\mathsf {T_{virt}}}}

\newenvironment{myabstract}
{\list{}{\listparindent 1.5em%
        \itemindent    \listparindent
        \leftmargin    1cm
        \rightmargin   1cm
        \parsep        0pt}%
    \item\relax}
{\endlist}

\newenvironment{mycover}
{\list{}{\listparindent 0pt
        \itemindent    \listparindent
        \leftmargin    1cm
        \rightmargin   1cm
        \parsep        0pt}%
    \raggedright
    \item\relax}
{\endlist}

\newcommand{\myemail}[1]{\,$\cdot$\, {\small #1}}
\newcommand{\myaff}[1]{\,$\cdot$\, {\small #1}\par\smallskip}

\hypersetup{
    colorlinks=true,
    linkcolor=black,
    citecolor=black,
    filecolor=black,
    urlcolor=[HTML]{0088CC},
    pdftitle={Locality of not-so-weak coloring},
    pdfauthor={Alkida Balliu, Juho Hirvonen, Christoph Lenzen, Dennis Olivetti, Jukka Suomela}
}

\begin{document}

\begin{mycover}
    {\huge\bfseries Locality of not-so-weak coloring \par}
    \bigskip
    \bigskip
    \bigskip

    \textbf{Alkida Balliu}
    \myemail{alkida.balliu@aalto.fi}
    \myaff{Aalto University}

    \textbf{Juho Hirvonen}
    \myemail{juho.hirvonen@aalto.fi}
    \myaff{Aalto University}

    \textbf{Christoph Lenzen}
    \myemail{clenzen@mpi-inf.mpg.de}
    \myaff{Max Planck Institute for Informatics}

    \textbf{Dennis Olivetti}
    \myemail{dennis.olivetti@aalto.fi}
    \myaff{Aalto University}

    \textbf{Jukka Suomela}
    \myemail{jukka.suomela@aalto.fi}
    \myaff{Aalto University}
\end{mycover}
\bigskip

\begin{myabstract}
\noindent\textbf{Abstract.}
Many graph problems are locally checkable: a solution is globally feasible if it looks valid in all constant-radius neighborhoods. This idea is formalized in the concept of \emph{locally checkable labelings} (\lcl{}s), introduced by Naor and Stockmeyer (1995). Recently, Chang et al.\ (2016) showed that in bounded-degree graphs, every \lcl{} problem belongs to one of the following classes:
\begin{itemize}[noitemsep]
    \item \emph{``Easy'':} solvable in $O(\log^* n)$ rounds with both deterministic and randomized distributed algorithms.
    \item \emph{``Hard'':} requires at least $\Omega(\log n)$ rounds with deterministic and $\Omega(\log \log n)$ rounds with randomized distributed algorithms.
\end{itemize}

Hence for any parameterized \lcl{} problem, when we move from local problems towards global problems, there is some point at which complexity suddenly jumps from easy to hard. For example, for vertex coloring in $d$-regular graphs it is now known that this jump is at precisely $d$ colors: coloring with $d+1$ colors is easy, while coloring with $d$ colors is hard.

However, it is currently poorly understood where this jump takes place when one looks at \emph{defective} colorings. To study this question, we define \emph{$k$-partial $c$-coloring} as follows: nodes are labeled with numbers between $1$ and $c$, and every node is incident to at least $k$ properly colored edges.

It is known that $1$-partial $2$-coloring (a.k.a.\ weak $2$-coloring) is easy for any $d \ge 1$. As our main result, we show that $k$-partial $2$-coloring becomes hard as soon as $k \ge 2$, no matter how large a $d$ we have.

We also show that this is fundamentally different from $k$-partial $3$-coloring: no matter which $k \ge 3$ we choose, the problem is always hard for $d = k$ but it becomes easy when $d \gg k$. The same was known previously for partial $c$-coloring with $c \ge 4$, but the case of $c < 4$ was open.
\end{myabstract}

\newpage

\section{Introduction}

\begin{table}
\raggedleft
\renewcommand{\arraystretch}{1.4}
\definecolor{cct}{HTML}{ccebc5}
\definecolor{cce}{HTML}{b3cde3}
\definecolor{cch}{HTML}{fbb4ae}
\colorlet{ct}{cct!45!white}
\colorlet{ce}{cce!45!white}
\colorlet{ch}{cch!35!white}
\colorlet{cu}{black!3!white}
\colorlet{cE}{cce!70!black}
\colorlet{cH}{cch!70!black}
\newcommand{\Sp}{\,}
\newcommand{\cd}[1]{\multicolumn{8}{@{}l@{}}{\emph{#1}}}
\newcommand{\ck}[1]{\multicolumn{1}{@{}r}{$#1$:$\ $}}
\newcommand{\cc}[1]{\multicolumn{2}{c}{$c=#1$}}
\newcommand{\ct}[1]{\multicolumn{2}{c}{\cellcolor{ct}$#1$}}
\newcommand{\ch}[1]{\multicolumn{2}{c}{\cellcolor{ch}$#1$}}
\newcommand{\cww}[1]{%
    \multicolumn{1}{>{\columncolor{cu}[\tabcolsep][0pt]}r@{}}{$\Sp#1$} &
    \multicolumn{1}{@{}>{\columncolor{cu}[0pt][\tabcolsep]}l}{${} \ldots \infty\Sp$}
}
\newcommand{\cee}[2]{%
    \multicolumn{1}{>{\columncolor{ce}[\tabcolsep][0pt]}r@{}}{$\Sp#1$} &
    \multicolumn{1}{@{}>{\columncolor{ce}[0pt][\tabcolsep]}l}{${} \ldots #2\Sp$}
}
\newcommand{\nnH}[1]{\color{cH}\fbox{\color{black}$#1$}}
\newcommand{\oaE}[1]{#1\hspace{\fboxsep}\hspace{\fboxrule}}
\newcommand{\obE}[1]{\hspace{\fboxsep}\hspace{\fboxrule}#1}
\newcommand{\nnE}[1]{\color{cE}\fbox{\color{black}$#1$}}
\newcommand{\Cee}[2]{\cee{\oaE{#1}}{\obE{#2}}}
\newcommand{\CeE}[2]{\cee{\oaE{#1}}{\nnE{#2}}}
\newcommand{\CEe}[2]{\cee{\nnE{#1}}{\obE{#2}}}
\newcommand{\CEE}[2]{\cee{\nnE{#1}}{\nnE{#2}}}
\newcommand{\cH}[1]{\ch{\nnH{#1}}}
\begin{tabular}{cccccccccccccccccc}
&\cd{(a) Before this work:} & \  &\cd{(b) After this work:} \\ \addlinespace
&\cc{2}&\cc{3}&\cc{4}&\cc{5}&&\cc{2}&\cc{3}&\cc{4}&\cc{5}\\
\ck {k=1} & \ct {1}      & \ct  {1}     &\ct{1}      & \ct{1} &
          & \ct {1}      & \ct  {1}     &\ct{1}      & \ct{1} \\
\ck   {2} & \cww{3}      & \ct  {2}     &\ct{2}      & \ct{2} &
          & \cH {\infty} & \ct  {2}     &\ct{2}      & \ct{2} \\
\ck   {3} & \cww{4}      & \cww {4}     &\ct{3}      & \ct{3} &
          & \cH {\infty} & \CeE {4}{5}  &\ct{3}      & \ct{3} \\
\ck   {4} & \cww{5}      & \cww {5}     &\cee{5}{7}  & \ct{4} &
          & \cH {\infty} & \CeE {5}{8}  &\CeE{5}{6}  & \ct{4} \\
\ck   {5} & \cww{6}      & \cww {6}     &\cee{6}{9}  & \cee{6}{9} &
          & \cH {\infty} & \CEE {7}{11} &\Cee{6}{9}  & \CeE{6}{7} \\
\ck   {6} & \cww{7}      & \cww {7}     &\cee{7}{11} & \cee{7}{11} &
          & \cH {\infty} & \CEE {8}{14} &\Cee{7}{11} & \Cee{7}{11} \\
\ck   {7} & \cww{8}      & \cww {8}     &\cee{8}{13} & \cee{8}{13} &
          & \cH {\infty} & \CEE{10}{17} &\CEe{9}{13} & \Cee{8}{13}
\end{tabular}
\caption{An overview of $k$-partial $c$-coloring in $d$-regular graphs: for each $k$ and $c$, the table shows what is the smallest $d$ such that the problem is easy. For example, ``$4\ldots5$'' means that for these parameters the problem is known to be easy in $5$-regular graphs, while the case of $4$-regular graphs is unknown. The new results are highlighted with a frame. Our main contributions are the new lower bounds for $c=2$ (Theorem~\ref{thm:lower_bound}) and upper bounds for $c=3$ (Theorem~\ref{thm:upper}), which were previously completely open. We also obtain stronger lower bounds for e.g.\ $c = 3$, $k \ge 5$ (Theorem~\ref{thm:lb-proper-col}) and stronger upper bounds for $c=k\ge 4$ (Theorem~\ref{thm:upper}).}\label{tab:overview}
\end{table}

There is a broad family of graph problems---so-called \emph{locally checkable labelings} or \lcl{}s \cite{Naor1995}---that exhibits the following dichotomy \cite{chang16exponential}: either the problem can be solved in $O(\log^* n)$ rounds with deterministic distributed algorithms, or any such algorithm requires at least $\Omega(\log n)$ rounds.

Hence, for any parameterized \lcl{} problem there is a sudden jump in complexity, from $O(\log^* n)$, which is a very slowly-growing function of $n$, to $\Omega(\log n)$, which can be already as much as the diameter of the network. We will call these two cases \emph{``easy''} and \emph{``hard''} from now on.

If we look at $d$-regular graphs for constant $d = O(1)$, then by prior work the following thresholds are known \cite{panconesi01simple,cole86deterministic,Brandt2016,chang16exponential}:
\begin{itemize}[noitemsep]
    \item Proper vertex coloring with $c$ colors: easy for $c \ge d+1$, hard for $c \le d$.
    \item Proper edge coloring with $c$ colors: easy for $c \ge 2d-1$, hard for $c \le 2d-2$.
\end{itemize}
Here the easy cases are exactly those cases that can be solved with a greedy algorithm that picks the colors of the nodes or edges one by one; a straightforward parallelization of this idea then gives an $O(\log^* n)$-round distributed algorithm.

In this work, we study colorings that are not necessarily proper:

\begin{definition}
Let $G=(V,E)$ be a graph. Mapping $f\colon V \to \{1,2,\dotsc,c\}$ is a $k$-partial $c$-coloring if for each node $v\in V$, there are at least $k$ neighbors $u$ of $v$ with $f(u) \ne f(v)$.
\end{definition}

By prior work on defective colorings, it is known that e.g.\ $k$-partial $4$-coloring is hard if $d=k\ge 4$ and easy if $d \gg k$. However, very little was known about partial $2$-coloring and $3$-coloring. In this work we complete the picture and show that the case of $c=2$ is very different from the case $c \ge 3$:
\begin{itemize}[noitemsep]
    \item $k$-partial $2$-coloring for any $k \ge 2$ is always hard, no matter how large a $d$ we have,
    \item $k$-partial $3$-coloring for any $k \ge 3$ is hard for $d = k$ but it becomes easy when $d \gg k$.
\end{itemize}
We summarize our contributions in Table~\ref{tab:overview}.

\section{Preliminaries and related work}

\subsection{\local{} model of computing}

We work in the usual \local{} model of distributed computing \cite{Linial1992,Peleg2000}. Each node of the input graph $G = (V,E)$ is a computer and each edge is a communication link. Computation proceeds in rounds: in one round each node can exchange a message (of any size) with each of its neighbors. Initially each node knows only $n = |V|$, and when a node stops, it has to produce its own part of the output---in our case, its own color from $\{1,2,\dotsc,c\}$. We say that an algorithm runs in time $T$ if after $T$ rounds all nodes stop and announce their local outputs.

When we study \emph{deterministic} algorithms, we assume that each node is labeled with a unique identifier from $\{1,2,\dotsc,n^{O(1)}\}$. When we study \emph{randomized} algorithms, we assume that each node has an unlimited source of random bits. For a randomized algorithm, we require that it solves the problem correctly \emph{with high probability,} i.e., with probability at least $1-n^{-c}$ for an arbitrary, but predetermined constant $c>0$.

Note that if a problem is solvable in $T$ rounds in the \local{} model, it also means that each node can produce its own part of the solution based on the information that is available in its radius-$T$ neighborhood.

\subsection{\lcl{} problems and gap theorems}

\lcl{} problems were introduced by \citet{Naor1995} in 1995. In an \lcl{} problem, the input is a graph $G = (V,E)$ of maximum degree $\Delta = O(1)$, possibly labeled with some node labels from a constant-size set $X$. The task is to find a labeling $f\colon V \to Y$, for some constant-size set $Y$, that satisfies some local constraints---a labeling is globally feasible if it is feasible in all constant-radius neighborhoods.

For our purposes it is enough to note that $k$-partial $c$-coloring in $d$-regular graphs is an \lcl{} problem, for any choice of constants $k, c, d = O(1)$. Hence also everything that we know about \lcl{}s applies here.

In the past four years, we have seen a lot of progress in understanding the computational complexity of \lcl{} problems in the \local{} model \cite{Balliu2018stoc,Balliu2018disc,Brandt2016,chang16exponential,chang17hierarchy,fischer17sublogarithmic,ghaffari17distributed,Ghaffari2018,Ghaffari2018a,chang18complexity,Brandt2017}. For us, the most relevant result is the gap theorem by \citet{chang16exponential}. They show that every \lcl{} problems belongs to one of the following classes, which we will here informally call ``easy'' and ``hard''
\begin{description}[noitemsep]
    \item[``Easy'':] solvable in $O(\log^* n)$ rounds with both deterministic and randomized algorithms.
    \item[``Hard'':] requires $\Omega(\log n)$ rounds with deterministic algorithms and $\Omega(\log \log n)$ rounds with randomized algorithms.
\end{description}

In this work, our main goal is to understand for what values of $k, c, d$ the problem of finding $k$-partial $c$-coloring in $d$-regular graphs is ``hard'' and when it is ``easy''. While we will focus on the case of $d$-regular graphs, most of the results directly generalize to the case of graphs of minimum degree $d$ (and maximum degree some constant $\Delta$).

\subsection{Prior work related to partial colorings}

\paragraph{Notes on terminology.}
In $d$-regular graphs, a $k$-partial $c$-coloring is exactly the same thing as a \emph{$(d-k)$-defective $c$-coloring} \cite[Sect.~6]{Barenboim2013}. While defective colorings are more commonly discussed in prior work, for our purposes the concept of a partial coloring is much more convenient, as we will often fix $c$ and $k$ and see what happens when $d$ increases.

Our definition is in essence equal to \emph{$k$-partially proper colorings} used by \citet{Kuhn2009}; for brevity, we call these partial colorings.

\paragraph{Weak coloring, \boldmath $k=1$.}
In graphs without isolated nodes, a $1$-partial $c$-coloring is identical to a \emph{weak $c$-coloring} \cite{Naor1995}. Weak $2$-coloring can be solved in $O(\log^* n)$ rounds: find a maximal independent set $X \subseteq V$ using e.g.\ \cite{panconesi01simple,cole86deterministic}; then color all nodes of $X$ with color $1$ and all other nodes with color $2$. Naturally, this also gives a solution for weak $c$-coloring for any $c \ge 2$. Furthermore, this upper bound is tight: a weak $2$-coloring breaks symmetry everywhere in a regular grid, and the usual lower bounds \cite{Linial1992,Naor1991,Naor1995} apply.

Weak coloring corresponds to the first row of Table~\ref{tab:overview}.

\paragraph{Partial coloring for \boldmath $k<c$.}
Above we have seen that we can find a $1$-partial $2$-coloring by simply finding a maximal independent set (MIS). The same idea can be generalized to {$(c-1)$}-partial $c$-coloring for $k = c-1$: Find an MIS $X_1$, label $X_1$ with color $1$, and remove $X_1$. Find an MIS $X_2$, label $X_2$ with color $2$, and remove $X_2$, etc. We continue this for $c-1$ steps and finally label all remaining nodes with color $c$.

The region where this simple (folklore?) strategy works is indicated with green color in Table~\ref{tab:overview}.

\paragraph{Proper vertex coloring, \boldmath $k=d$.}
In $d$-regular graphs, a $d$-partial $c$-coloring is a proper $c$-coloring. Recall that proper coloring with $d+1$ colors is easy \cite{panconesi01simple,cole86deterministic}, while proper coloring with $d$ colors is known to be hard \cite{Brandt2016,chang16exponential}.

Hardness of proper $d$-coloring implies the lower bounds in the blue and gray regions of Table~\ref{tab:overview}a.

\paragraph{Partial coloring for \boldmath $c \ge 4$.}
\citet{barenboim14distributed} gave an algorithm that computes a $\lfloor d/p \rfloor$-defective $p^2$-coloring in time $O(\log^* n)$, which is essentially a defective variant of Linial's $O(\Delta^2)$-coloring algorithm \cite{Linial1992}. This algorithm requires at least $4$ colors, and for the case $c=4$ it translates to a $\lceil d/2 \rceil$-partial $4$-coloring. For example, $4$-partial $4$-coloring is therefore easy in $7$-regular graphs, and $5$-partial $4$-coloring is easy in $9$-regular graphs.

This algorithm gives the upper bounds in the blue region of Table~\ref{tab:overview}a.

\paragraph{Partial coloring for \boldmath $c \le 3$.}
To our knowledge, no $O(\log^* n)$-time algorithms are known for $k$-partial $c$-coloring for $c \le 3$, $k \ge c$. In particular, it is not known if the problem becomes easy in $d$-regular graphs for sufficiently large values of $d \gg k$.

This unknown region is indicated with a gray shading in Table~\ref{tab:overview}a.

\paragraph{Algorithms based on Lov\'asz local lemma.}
\citet{Chung2017}, \citet{fischer17sublogarithmic}, and \citet[full version]{Ghaffari2018} present algorithms for defective coloring (and hence for partial coloring) that are based on the following idea: formulate a defective coloring as an instance of the Lov\'asz local lemma (LLL), and then apply efficient distributed algorithms for LLL.

Unfortunately, this approach is unlikely to lead to an $O(\log^* n)$-time algorithm; LLL is known to be a hard problem for a wide range of parameters \cite{Brandt2016}.

\paragraph{Other algorithms.}
\citet{Bonamy2018} shows that there is an $O(\log n)$-round algorithm for trees that finds an MIS $X$ such that every component induced by non-MIS nodes is of size one or two. This can be interpreted as an algorithm for partial $2$-coloring.

However, this approach cannot lead to an $O(\log^* n)$-time algorithm, either: if we color the MIS-nodes with color $1$ and the non-MIS nodes with colors $2$ and $3$, we obtain a proper $3$-coloring, and finding a $3$-coloring in $3$-regular trees is known to be a hard problem \cite{Brandt2016}.

\section{Our contributions}

To recap, by prior work, we have a good qualitative understanding of $k$-partial $c$-coloring for $c \ge 4$:
\begin{itemize}[noitemsep]
    \item $k < c$: easy for all $d \ge k$.
    \item $k \ge c$: hard for $d = k$ but easy for $d \gg k$.
\end{itemize}
We complete the picture for $c \le 3$. For $c = 3$, we have precisely the same situation as above:
\begin{itemize}[noitemsep]
    \item $k < c$: easy for all $d \ge k$.
    \item $k \ge c$: hard for $d = k$ but easy for $d \gg k$.
\end{itemize}
However, the case of $c = 2$ is fundamentally different:
\begin{itemize}[noitemsep]
    \item $k < c$: easy for all $d \ge k$.
    \item $k \ge c$: hard for all values of $d$.
\end{itemize}

\subsection{Corollary: locally optimal cuts}

Any partial $2$-coloring can be interpreted as a \emph{cut}; the properly colored edges are \emph{cut edges}, and the \emph{size} of the cut is  the number of cut edges. Let us look at the problem of maximizing the size of a cut with a simple greedy strategy: start with any cut and change the color of a node if it increases the size of the cut. The process will converge to a \emph{locally optimal cut}, in which changing the color of any single node does not help.

Now a locally optimal cut in $d$-regular graphs is precisely the same thing as a $\lceil d/2 \rceil$-partial $2$-coloring. For example, in $3$-regular graphs, any $2$-partial $2$-coloring is also a locally optimal cut, and vice versa.

Locally optimal cuts are easy to find in a centralized, sequential setting. However, previously it was not known if locally optimal cuts can be found efficiently in a distributed setting. As a corollary of our work, we now know that this is a hard problem.

\subsection{Key techniques}
\paragraph{Upper bound for 3-coloring.}
Prior algorithms for e.g.\ partial $4$-coloring are based on the idea of organizing nodes in layers and doing two sweeps \cite{barenboim14distributed}: top to bottom, using colors from palette $A = \{1,2\}$, and bottom to top using colors from palette $B = \{1,2\}$. This way we eventually have a $4$-coloring with colors from $A \times B = \{(1,1), (1,2), (2,1), (2,2)\}$. This idea generalizes easily to e.g.\ $6, 8, 9, \dotsc$ colors, but it is not possible to use this idea to find a useful coloring with less than $4$ colors.

We show how to do two sweeps so that the end result is only $3$ colors. In brief, the first sweep uses \emph{tentative} colors from palette $\{1,2\}$, and the second sweep \emph{finalizes} the colors, depending on the tentative colors that we chose in the first step. Here the second sweep depends on the result of the first sweep, while in prior algorithms the two sweeps are independent.
\paragraph{Lower bound for 2-coloring.}
We show that $2$-partial $2$-coloring in $d$-regular graphs for any constant $d$ is at least as hard to solve as \emph{sinkless orientation}, which is known to be hard \cite{Brandt2016}. The key obstacle here is that sinkless orientation is known to be hard even if we are given a \emph{proper} $2$-coloring of the graph, so how could a \emph{partial} $2$-coloring help with it?

The basic idea is as follows: Assume we have a fast algorithm $A_1$ that finds a $2$-partial $2$-coloring in $d_1$-regular graphs. Then we can construct algorithm $A_2$ that finds a sinkless orientation in $d_2$-regular graphs, for a certain constant $d_2 \gg d_1$ that depends on the exact running time of $A_1$. Given a $d_2$-regular graph $G_2$, algorithm $A_2$ first replaces all nodes with appropriate gadgets to obtain a $d_1$-regular graph $G_1$, applies $A_1$ to $G_1$, and extracts enough information from the partial coloring so that it can find a sinkless orientation. But sinkless orientation is hard also in $d_2$-regular graphs, no matter how large a constant $d_2$ is.

\section{Partial colorings with more than two colors}

In this section we analyze the distributed complexity of $k$-partial $c$-coloring in the case where $c$ is at least $3$. More formally, we will prove the following theorem.

\begin{theorem}\label{thm:upper}
    There exists an algorithm running in $O(\log^* n)$ that is able to compute:
    \begin{itemize}[noitemsep]
        \item A $k$-partial $3$-coloring, if $d \ge 3k-4$ and $k \ge 3$;
        \item A $k$-partial $k$-coloring, if $d \ge k+2$ and $k \ge 4$.
    \end{itemize}
\end{theorem}

In order to prove the theorem, we start by providing an algorithm, and then we will analyze it for the two cases separately.

\paragraph{The algorithm.}
The algorithm that we propose is inspired by the procedure \emph{Refine} in \cite[Sect.~6]{Barenboim2013}. This procedure starts by first finding an acyclic partial orientation, and then assigns two colors for each node by exploiting the two possible orders given by the orientation. It finally combines the two colors to determine the output color. Our algorithm starts in the same way, but it does \emph{not} compute two independent colors, allowing us to be slightly more efficient in some cases.

We start by finding an acyclic partial orientation of the edges. That is, we first compute an $O(d^2)$ coloring in $O(\log^* n)$ rounds. Then, we assign a total order to the colors, and we orient the edges from the node with the smaller color to the node with the bigger one. The obtained directed graph is clearly acyclic, and all directed paths are of length at most $O(d^2)$. Nodes reachable from $v$ through outgoing edges are considered to be \emph{above} $v$, while the others are considered to be \emph{below} $v$.

Now, we do two ``sweeps'' on the obtained acyclic graph, that is, we first process the nodes from above to below, and then we process them in the reverse order. More precisely, we start by processing the sinks, and then we continue by processing all nodes such that all of the nodes above them have already been processed. This is iterated until all nodes have been processed. Then, we repeat the same procedure in reverse order, i.e., from below to above. Each sweep takes $O(d^2)$ rounds.

During the first sweep, we assign to each node $v$ a temporary color, by choosing the color that is the least used one among the neighbors above $v$. Crucially, during this phase, the choice is not among the full palette, but only colors from $1$ to $c-1$ are allowed. We call color $c$ the \emph{special} color. During the second sweep, each node $v$ has three options:
\begin{itemize}
    \item Keep the current color.
    \item Choose to switch to color $c$.
    \item Choose to switch to a color from $1$ to $c-1$. This option is allowed only if no neighbor below $v$ is currently using that color.
\end{itemize}

\begin{figure}
    \centering
    \includegraphics[scale=0.9]{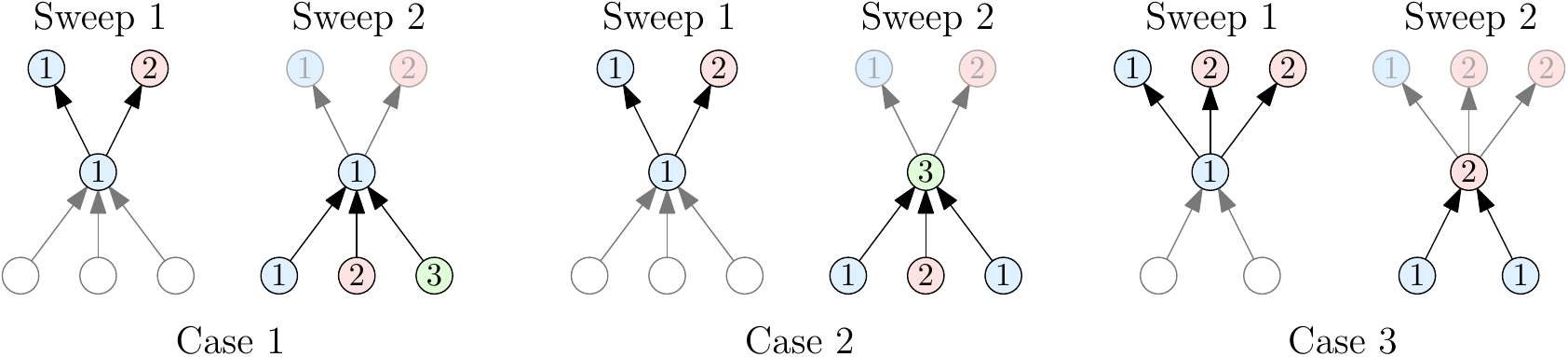}
    \caption{Examples of the output of the $3$-partial $3$-coloring algorithm, running at the central node $v$, in a graph where the degree is $5$. The figure shows 3 cases: $v$ keeps the color chosen during the first sweep, $v$ switches to the special color in the second sweep, $v$ switches to color $2$ in the second sweep.}\label{fig:algo}
\end{figure}

Different choices give different guarantees. For example, by choosing to not change the color, or by choosing to switch to a color from $1$ to $c-1$, node $v$ is ensured that the number of properly colored neighbors does not decrease when the nodes above it will be processed. This property is guaranteed by the fact that a node can switch to a color in $\{1,\ldots,c-1\}$ only if no node below it is using that color. On the other hand, a node may switch to color $c$ even if some neighbor below it is using color $c$ as well, but then it loses any guarantee about the nodes above it, which may all decide to switch to color $c$. See Figure \ref{fig:algo} for an example of the execution of the algorithm.

The algorithm will make a choice that guarantees that $k$ neighbors have different color, regardless of whether the nodes above change their color (subject to the above rules) or not. Accordingly, our task is to prove that such a choice always exists, provided that $d$ is large enough.

\paragraph{\boldmath $k$-partial 3-coloring.}
We now show that the above described algorithm is able to compute a $k$-partial $3$-coloring if $d \ge 3k-4$ and $k \geq 3$. In order to analyze the algorithm running on node $v$, assume without loss of generality that during the first sweep, $v$ picks color $1$ and $t$ nodes above $v$ chose color $2$. Note that there are no more than $t$ other nodes above $v$, as $v$ picked the color out of $\{1,2\}$ that was least used by the nodes above it in the first sweep. Denote by $x,y,z\geq 0$ the numbers of nodes below $v$ that are colored $1$, $2$, and $3$, respectively, after making their final choice in the second sweep. Thus, the number of nodes above $v$ that chose $1$ in their first sweep equals $d-t-x-y-z\leq t$.

We make a case distinction.
\begin{enumerate}
  \item $t+y+z\geq k$. Thus, $v$ can keep color $1$, as the $t$ nodes above it that have color $2$ must then choose a color different from $1$.
  \item $x+y\geq k$. Then $v$ can safely choose color $3$.
  \item $y=0$ and none of the other cases apply. Thus, $v$ is free to switch to color $2$. If it does so, it has $x+z$ nodes of different color below, and $d-t-x-z$ nodes above that choose a different color than $2$. As the first case does not apply and $k\geq 3$, these are at least $d-t\geq d - (k-1)\geq 2k-3 \geq k$ nodes. Hence, switching to color $2$ is indeed a valid choice.
\end{enumerate}
Hence, it suffices to show that one of the cases must apply. Assume for contradiction that this is false. Thus,
\begin{align*}
t+y+z &\le k-1\,,\\
x+y &\le k-1\,,\\
y &\ge 1\,,\\
\mbox{and}\qquad \qquad d&\le 2t+x+y+z\,,
\end{align*}
yielding the contradiction that
\begin{equation*}
d\leq 2t+x+y+z\leq 2(k-1)-y-z+x \leq 3(k-1)-2y-z\leq 3k-5\,.
\end{equation*}

\paragraph{\boldmath $k$-partial $k$-coloring.}
We now show that the above described algorithm is able to compute a $k$-partial $k$-coloring if $d \ge k+2$ and $k \ge 4$. We analyze this case similarly to the case before. Let $t$ be the number of nodes above $v$ of color different from $v$ after the first sweep. Without loss of generality, assume that the color of $v$ is $1$, and the special color is $k$. Let $x$, $y$, and $z$ be the number of nodes below $v$ colored with $1$, with color $c$ such that $2 \le c \le k-1$, and with color $k$, respectively. Recall that only colors from $1$ to $k-1$ are allowed during the first sweep. As $v$ chooses a minority color among its above neighbors' choices, we have $r:=d-t-x-y-z\leq t/(k-2)$ remaining nodes above $v$ that choose color $1$ in the first sweep.

Let us analyze the second sweep. We make a case distinction.
\begin{enumerate}
  \item $t+y+z\geq k$. Then $v$ can keep color $1$.
  \item $x+y\geq k$. Then $v$ can safely choose color $k$.
  \item There are $f>0$ ``free'' colors from $\{2,\ldots,k-1\}$ that no neighbor below $v$ chose and none of the other cases applies. If a free color was picked by at most $d-k$ above neighbors of $v$, it may select it with the guarantee that its other $k$ neighbors end up with a different color.
  
  Assuming for contradiction that there is no such free color, observe that the least used free color was picked by at most $\lfloor t/f\rfloor$ neighbors above $v$ in the first sweep. Accordingly, $t\geq f(d-k+1)\geq d-k+f$. Moreover, clearly $f\geq k-2-y$ and, because the first case does not apply, $x+r=d-t-y-z\geq d-(k-1)\geq 3$. Thus, we can lower bound the total number of neighbors of $v$ by
  \begin{equation*}
  d = r + t + x + y + z \geq d-k+f+y+3\geq d+1\,,
  \end{equation*}
  a contradiction. Therefore, one of the free colors is a valid choice for $v$.
\end{enumerate}
Hence, there is indeed always a valid choice if we can show that the above case distinction is exhaustive. Assuming otherwise, collecting inequalities from the cases and the earlier bound on $r$ we obtain that
\begin{align*}
t+y+z&\leq k-1\,,\\
x+y&\leq k-1\,,\\
y&\geq k-2\,,\\
\mbox{and}\qquad \qquad r&\leq \frac{t}{k-2}\,.
\end{align*}
Together, this implies
\begin{equation*}
\begin{split}
(k-2)d &\leq (k-1)(t+x+y+z)-y\leq (k-1)^2+(k-1)(x+y)-ky \\ &\leq 2(k-1)^2-k(k-2)=k^2 - 2k +2\,,
\end{split}
\end{equation*}
yielding the contradiction that $d \leq k+2/(k-2)<k+2$ (using that $k\geq 4$).

\section{Two-partial two-coloring}
In this section, we show that, in the \local{} model, $2$-partial $2$-coloring requires $\Omega(\log n)$ deterministic time and $\Omega(\log \log n)$ randomized time in any $d$-regular tree, where $d\ge 2$. We show the result by reducing from the sinkless orientation problem, for which we know that its distributed deterministic complexity is $\omega(\log^* n)$ rounds in the \local{} model.

Informally, the proof proceeds in two steps. We first show that, if we can solve 2-partial 2-coloring in constant time in a slightly modified version of the \local{} model, which we call \dclocal{} model, then we can solve sinkless orientation in the \local{} model in $O(\log^* n)$ rounds. Subsequently, we show that, if we can solve an \lcl{} problem $P$ in $O(\log^* n)$ rounds in the \local{} model, then we can solve $P$ in the \dclocal{} model in constant time using a simulation similar to that of \citet{chang16exponential}. Therefore no $O(\log^* n)$-round algorithm exists, i.e., the problem is not easy and hence it has to be hard, i.e., it requires $\Omega(\log n)$ deterministic time and $\Omega(\log \log n)$ randomized time.

\begin{theorem}\label{thm:lower_bound}
    Computing a $2$-partial $2$-coloring in $d$-regular trees in the \local{} model requires $\Omega(\log n)$ deterministic time and $\Omega(\log \log n)$ randomized time, for any $d \ge 2$.
\end{theorem}

\paragraph{\boldmath\dclocal{} model.}
Consider the usual \local{} model with the following modification. Instead of having unique identifiers, nodes are given as input a color from a palette of $c$ colors, and this coloring of the nodes guarantees a distance-$k$ coloring of the graph. In other words, each node sees different colors in its distance-$k$ ball, but it may see repeated colors in its distance-$(k+1)$ ball.  We call this model \dclocal$(k,c)$ (where $\sf{DC}$ stands for distance coloring).

\subsection{\boldmath A lower bound for the \dclocal{} model}

In this section, we show that $2$-partial $2$-coloring in $d$-regular trees is not solvable in $k=O(1)$ rounds in the \dclocal$(k+1,d^{2(k+1)})$ model. We show this by reducing from the sinkless orientation problem in $2$-colored trees in the \local{} model. More precisely, we show that, if there is an algorithm $A$ solving $2$-partial $2$-coloring in time $k=O(1)$ in the \dclocal{} model, then we can use it to design an $O(\log^* n)$-round algorithm that solves sinkless orientation in $2$-colored trees in the \local{} model. This would give an $\omega(1)$ lower bound for $2$-partial $2$-coloring in the \dclocal{} model, since we know that sinkless orientation requires $\Omega(\log n)$ rounds, even in $2$-colored trees, in the \local{} model \cite{chang18complexity}.

\begin{figure}[t]
    \centering
    \includegraphics[scale=0.7]{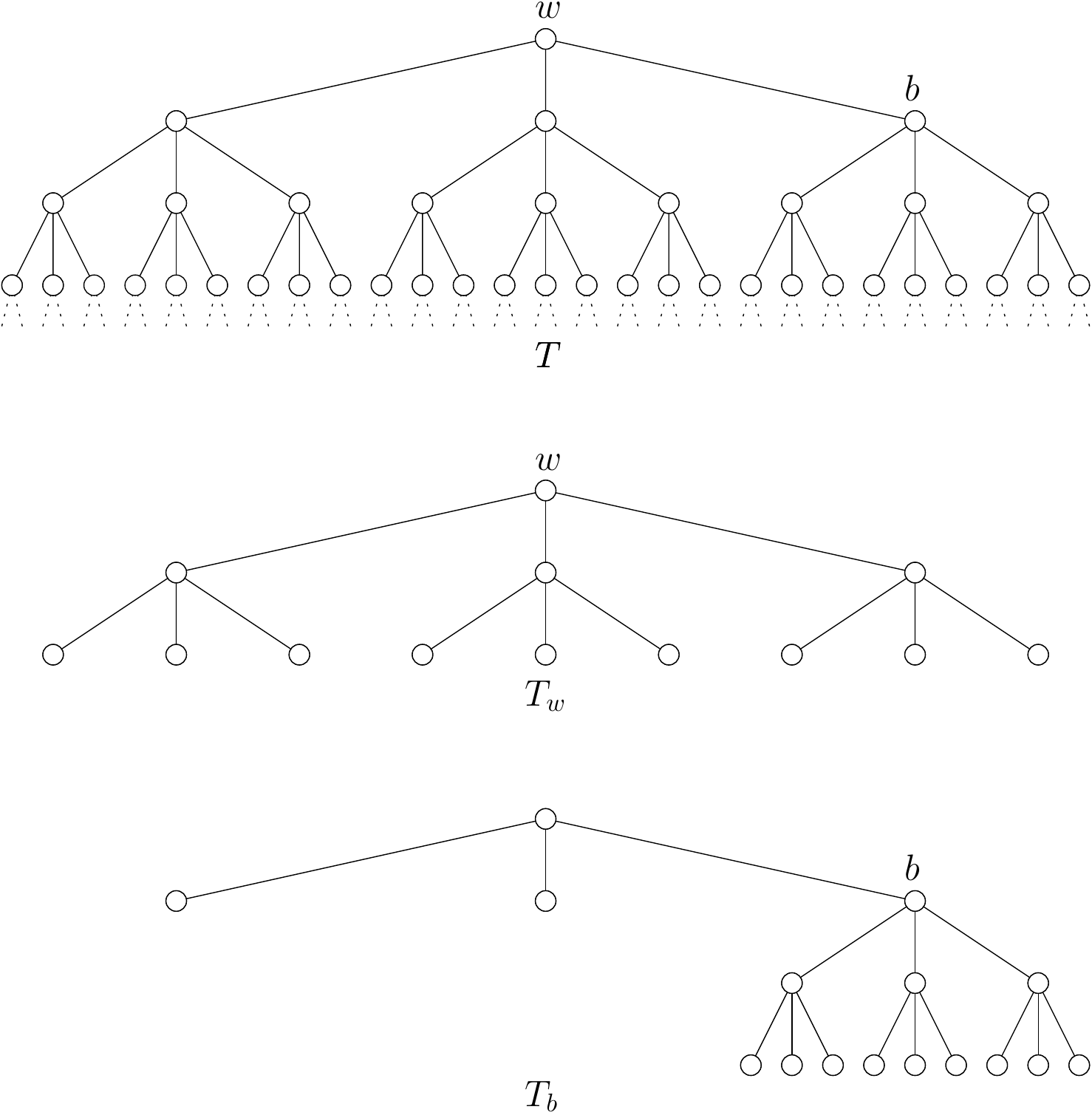}
    \caption{In this example $k=2$. Algorithm $A$ outputs ``white'' on the root of $T$ and ``black'' on the leftmost child of the root. The gadgets $T_w$ and $T_b$ are trees with depth $2$ rooted at $w$ and $b$ respectively.}\label{fig:gadget}
\end{figure}

\paragraph{Gadgets.}
Let $A$ be the algorithm that solves $2$-partial $2$-coloring in time $k=O(1)$ rounds in $d$-regular trees in the \dclocal$\bigl(k+1,d^{2(k+1)}\bigr)$ model. We introduce two gadgets that we will use later for proving the lower bound. Let $T$ be an arbitrarily distance-$(k+1)$ colored $d$-regular tree of depth $k+3$, and let $u$ be the root of $T$.
We run algorithm $A$ on $u$ and on each of its neighbors $v\in N(u)$. We denote with $A(v)$ the output of algorithm $A$ on a node $v$. Notice that $A(v)$ is well-defined on these nodes, since their $k$-radius ball is properly distance-$k$ colored and fully contained in $T$. Since algorithm $A$ solves $2$-partial $2$-coloring, we are sure that there are two nodes $v,z \in N(u)$ such that $A(u)\neq A(v)=A(z)$. Let $b\in \{u,v,z\}$ be a node for which $A$ outputs ``black'' and $w\in \{u,v,z\}$ be a node for which $A$ outputs ``white''. Let $T_w$ and $T_b$ be the subtrees of depth $k$ rooted at $w$ and $b$ respectively: these are our gadgets (see Figure \ref{fig:gadget} for an example). The gadgets $T_w$ and $T_b$ satisfy the following property.
\begin{property}\label{prop:gadget}
    Let $c\in\{\mbox{black, white}\}$ be the color of the root of the gadget and let $\bar{c}$ be the opposite color. Among all nodes at distance $2t$ from the root of the gadget, there must be at least one node for which the algorithm outputs color $c$.
\end{property}
\begin{proof}
    Assume that $A$ outputs $\bar{c}$ at all nodes at distance $2t$ from the root. Then all nodes at distance $2t-1$ must have color $c$ in order to guarantee a $2$-proper $2$-coloring. This would imply that all nodes at distance $2t-2$ have color $\bar{c}$. By applying this reasoning recursively, we conclude that the root must have color $\bar{c}$, which is a contradiction.
\end{proof}

\paragraph{Reduction.}
We now show that if there exists an algorithm $A$ that solves $2$-partial $2$-coloring in time $k=O(1)$ rounds (where $k$ is even) in $d$-regular trees in the \dclocal$\bigl(k+1,d^{2(k+1)}\bigr)$ model, then we can design an algorithm $A'$ that solves sinkless orientation on trees of degree $\Delta=d^{2k}$ in which a $2$-coloring of the tree is given, in $O(\log^* n)$ rounds in the \local{} model.

Consider a $\Delta$-regular $2$-colored tree $B=(V\cup U, E)$, where $V$ and $U$ represent the set of nodes belonging to the two color classes. We construct a virtual tree in the following way. Each node $x\in V\cup U$ pretends to be the root of a $d$-regular tree of depth $2k$. We call this tree \tvirt$(x)$. Then, each node $v \in V$ (resp. $u \in U$) labels the nodes at distance at most $k$ with the same colors of the nodes of the gadget $T_w$ (resp. $T_b$). Note that this is possible since $T_w$ and $T_b$ are isomorphic to the subgraph induced by the nodes at distance at most $k$ from $v$ and $u$. Then, we merge the $i$-th leaf of \tvirt$(v)$ with the $j$-th leaf of \tvirt$(u)$ if and only if there is an edge $\{v,u\}$ connecting $v$ and $u$ through port $i$ of $v$ and port $j$ of $u$. We call this node $m_{u,v}$ (merged node). In order to make the graph $d$-regular, we attach additional $d-2$ virtual nodes to each merged node. Note that, since the original tree $B$ is $\Delta$-regular, and since each virtual tree \tvirt{} has $\Delta$ leaves, then all the leaves of \tvirt{} are merged nodes (except for the case in which the original node is a leaf, where just one leaf of \tvirt{} is merged). We then color the nodes that are still uncolored, that is, nodes at distance more than $k$ from the root of the virtual trees, using a distance-$k$ coloring algorithm. Since the already colored parts are far enough apart, this can be done efficiently, in $O(\log^* n)$ rounds.

\begin{figure}
    \centering
    \includegraphics[scale=0.65]{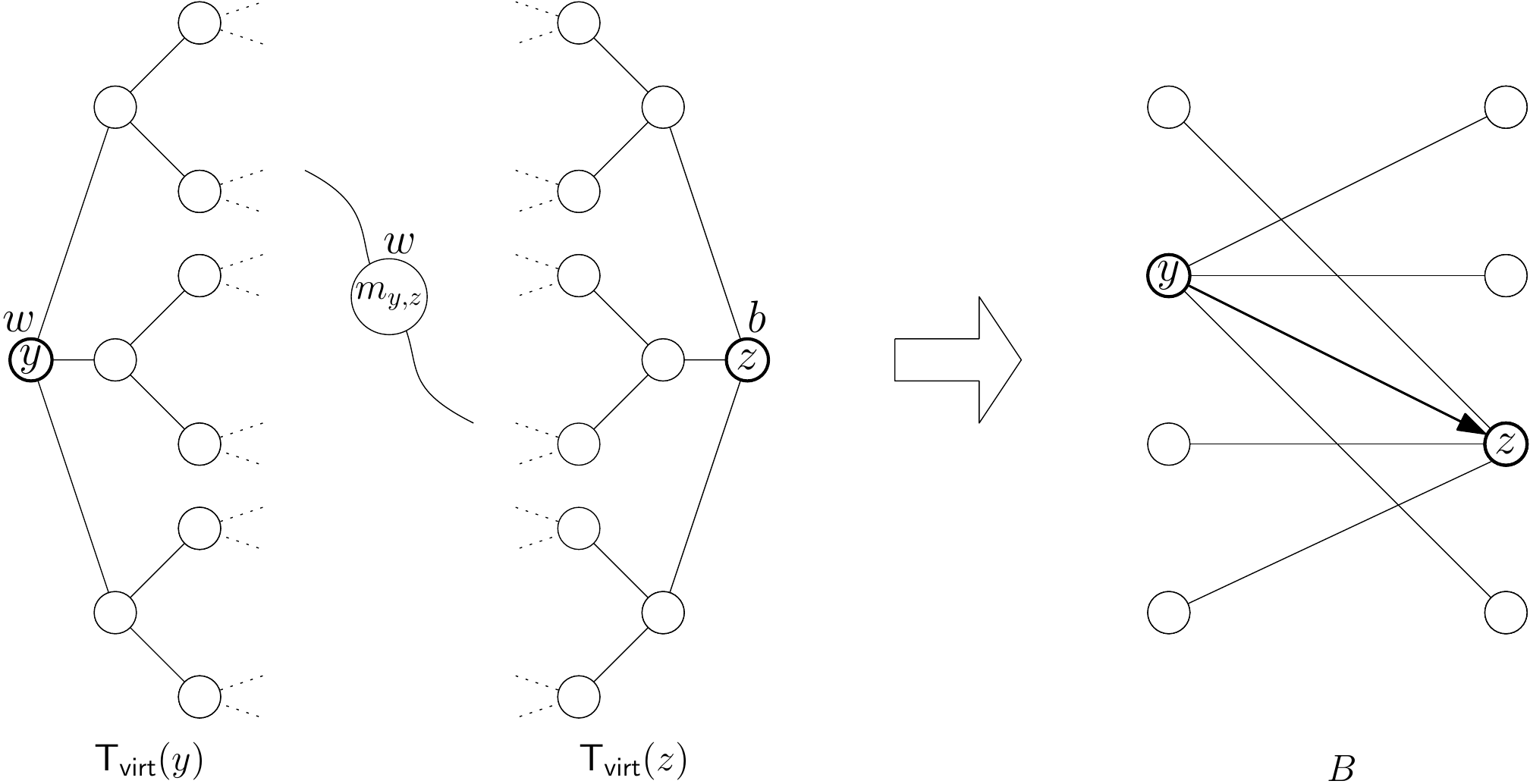}
    \caption{In this example, the merged node of $\tvirt(y)$ and $\tvirt(z)$ has the same color as node $y$, hence the edge $\{y,z\}$ in $B$ is oriented from $y$ to $z$.}\label{fig:virtual-graph}
\end{figure}

Now, each node $v$ in $B$ simulates the $k$-round algorithm $A$ on all nodes of $\tvirt(v)$ and gets a color $c\in\{ \mbox{black, white} \}$ for each node. This requires constant time. Algorithm $A$ outputs ``white'' at all nodes $v\in V$, since they have the same view as the root of $T_w$ up to distance $k$. Similarly, algorithm $A$ outputs ``black'' at all nodes $u\in U$. We then orient an edge of $B$ from node $y$ to node $z$ if and only if $m_{y,z}$ has the same color of $y$ (see Figure \ref{fig:virtual-graph} for an example). By Proposition \ref{prop:gadget}, we know that at least one such a leaf exists for each node in $B$, meaning that each node is guaranteed to have at least one outgoing edge. This would solve sinkless orientation on $2$-colored trees in the \local{} model in $O(\log^* n)$ rounds. Putting all together, we get the following lemma.

\begin{lemma}\label{thm:lb-in-DC}
    The 2-partial 2-coloring problem requires $\omega(1)$ rounds in the \dclocal{} model.
\end{lemma}

\subsection{\boldmath From \dclocal{} to \local}

We now show that all \lcl{}s solvable in $O(\log^*n)$ rounds can be solved in a standard manner, that is, first find a distance-$k$ coloring, for some constant $k$, and then apply a constant time algorithm running in at most $k$ rounds. In particular, we will prove the following lemma.
\begin{lemma}\label{thm:local-to-dclocal}
    Any \lcl{} problem that can be solved in $O(\log^* n)$ rounds in the \local{} model can be solved in $O(1)$ rounds in the \dclocal{} model.
\end{lemma}

We prove the lemma by simulation. Let $P$ be an \lcl{} problem checkable in $r$ rounds, where $r$ is some constant. Assume that we have an algorithm $A$ for the \local{} model that solves $P$ in $f(n) = O(\log^*n)$ rounds on graphs of size $n$ in which nodes have unique identifiers in $\{1,\ldots,n\}$. Let $\Delta = O(1)$ be the maximum degree of the graph.

Fix $N$ to be the smallest integer such that $t = f(N)+1$ and $\Delta^{2(t + r)} < N$. We show that we can construct an algorithm $A'$ running in $t$ rounds that solves $P$ in the \dclocal$\bigl(t+r$, $\Delta^{2(t+r)}\bigr)$ model. Note that $t$ is constant. In other words, we design an algorithm that solves $P$ in constant time given the promise that nodes are labeled with a $\Delta^{2(t+r)}$-coloring of distance $(t+r)$. We assume that the diameter of the graph is at least $2t$, otherwise nodes can gather the entire graph in constant time and solve $P$ by brute force.

Algorithm $A'$ executed by a node $v$ is defined as follows.  At first, node $v$ gathers its distance-$t$ neighborhood $B_v(t)$. Then, node $v$ creates a virtual instance of $P$ by renaming the nodes in $B_v(t)$ and setting their identifiers as their assigned colors. Now, node $v$ simulates algorithm $A$ on $B_v(t)$, by lying about the size of the graph and setting it to be $N$. Finally, the output of $A'$ is defined to be the same as the output of algorithm $A$. Notice that this simulation is clearly possible, since $A$, running on instances of size $N$, terminates in strictly less than $t$ rounds.

We still need to show that the output is valid for the original \lcl{}. For this purpose, we show that, if the algorithm fails in some neighborhood, then we can construct an instance in which the original algorithm fails as well. Let $G$ be the graph in which, given a $\Delta^{2(t+r)}$-coloring of distance $(t+r)$, there is a node $v$ such that the verifier executed on $v$ rejects (after running for $r$ rounds). Consider $G' = B_v(t+r)$, the subgraph of radius $r+t$ centered at $v$. All nodes in $G'$ have different colors and the number of nodes is at most $N$, since $N$ satisfies $\Delta^{2(t + r)} < N$. We now modify $G'$ in order to make it a graph of size exactly $N$. For this purpose, we pick an arbitrary node at distance $t+r$ from $v$ (that exists by the diameter assumption), and we connect to it a path of as many nodes as needed. We then complete the coloring of these nodes in some consistent manner.

The identifiers of nodes in $G'$ are set to be equal to their colors. The ID space in $G'$ is in $1,\ldots,N$. At this point, we run algorithm $A$ on $G'$. Consider the set $S$ of nodes at distance at most $r$ from $v$. For every node $u\in S$, the $t$-neighborhood of $u$ is the same on $G$ and $G'$, hence the output of $A$ on these nodes must be the same as the output of $A'$. Thus, the failure of $A'$ on $G$ implies the failure of $A$ on $G'$.
Theorem \ref{thm:lower_bound} follows by combining Lemmas \ref{thm:lb-in-DC} and \ref{thm:local-to-dclocal}.

\section{Additional hardness results}

\begin{theorem} \label{thm:lb-proper-col}
Computing a $k$-partial $c$-coloring in $d$-regular graphs, for $k \ge \frac{c-1}{c}d + 1$, requires $\Omega(\log n)$ deterministic time and $\Omega(\log \log n)$ randomized time.
\end{theorem}
\begin{proof}
Assume the problem is easy to solve. Each monochromatic subgraph has a maximum degree $x = d - k \le \frac{d}{c} - 1$, and hence it is easy to color with $x + 1 \le \frac{d}{c}$ colors. Hence overall we can easily find a proper coloring of a $d$-regular graph with at most $c \cdot \frac{d}{c} = d$ colors, but this is known to be hard \cite{Brandt2016,chang16exponential}.
\end{proof}

\section*{Acknowledgments}

This work was supported in part by the Academy of Finland, Grants 285721 and 314888.

{
    \urlstyle{sf}
    \DeclareUrlCommand{\Doi}{\urlstyle{same}}
    \renewcommand{\doi}[1]{\href{http://dx.doi.org/#1}{\footnotesize\sf doi:\Doi{#1}}}
    \bibliographystyle{plainnat}
    \bibliography{partial-coloring}
}

\end{document}